\documentclass[11pt]{article}
\usepackage{fullpage,amsmath,amssymb,amsfonts,amsthm}
\usepackage{tikz}
\usetikzlibrary{shapes}
\usepackage{a4wide}
\usepackage{caption}
\usepackage{titling}
\thanksmarkseries{arabic}
\usepackage{cite}

\newtheorem{defi}{Definition}
\newtheorem{prop}{Proposition}
\newtheorem{coro}{Corollary}
\newtheorem{lemma}{Lemma}

\DeclareMathOperator{\Tr}{Tr}

\newcommand{\R}{\mathbb{R}}
\newcommand{\C}{\mathbb{C}}
\newcommand{\Z}{\mathbb{Z}}

\newcommand{\h}{\mathcal{H}}

\newcommand\ket[1]{|#1\rangle}

\newcommand\proj[1]{|#1\rangle \langle #1|}

\title{Equivalence between contextuality and negativity of the Wigner function for qudits}

\author{
Nicolas Delfosse
	\thanks{Department of Physics and Astronomy, University of California, Riverside, CA, USA}
	$^{,}$\thanks{IQIM, California Institute of Technology, Pasadena, CA, USA} \ \  
Cihan Okay \thanks{Department of Mathematics, University of Western Ontario, London, Ontario, Canada} \ \  
Juan Bermejo-Vega \thanks{Dahlem Center for Complex Quantum Systems, Freie Universit{\"a}t Berlin, Berlin, Germany} \\
Dan E. Browne \thanks{Department of Physics and Astronomy, University College London, Gower Street, London, UK} \ \ and
Robert Raussendorf \thanks{Department of Physics and Astronomy, University of British
Columbia, Vancouver, BC, Canada}
	}

\begin{document}

\maketitle

\begin{abstract}
Contextuality and negativity of the Wigner function are two notions of non-classicality for quantum 
systems. Howard, Wallman, Veitch and Emerson proved recently that these two notions coincide
for qudits in odd prime dimension. This equivalence is particularly important since it promotes contextuality as 
a ressource that magic states must possess in order to allow for a quantum speed-up. 
We propose a simple proof of the equivalence between contextuality and negativity of the 
Wigner function based on character theory. This simplified approach allows us to generalize this equivalence to multiple qudits and to any qudit system of odd local dimension.
\end{abstract}

\section{Introduction}

Understanding what distinguishes quantum mechanics from classical mechanics
and probabilistic models is a central question of physics.
Besides its fondational aspects, this question is crucial for quantum information
processing applications since the features that set quantum and classical mechanics appart
are precisely the properties that we must exploit in order
to obtain a quantum superiority for certain tasks 
\cite{
	Cleve1998:quantum_algo, 
	vidal2003:entanglement, 
	terhal2004:constant_depth,
	vandennest2007:simulation_MBQC, 
	gross2009:overentaglement, 
	anders2009:correlations, 
	galvao2005:DWF_speedup,
	aharonov2006:QFT_simulation,
	yoran2007:QFT_simulation,
	browne2007:QFT,
	markov2008:tensor_network_simulation,
	yoran2008:entanglement,
	vandennest2011:proba_simulation, 
	veitch2012:negativity, 
	vandennest2013:disentangled_universality, 
	raussendorf2013:contextuality_MBQC, 
	howard2014:contextuality, 
	stahlke2014:interference
}.
In the present work, we compare two notions of non-classicality:
Contextuality 
\cite{bell1966:hvm, kochen1975:hvm, abramsky2011:sheaf_contextuality, acin2015:graph_contextuality} 
and negativity of the Wigner function \cite{wigner1932:WF, gibbons2004:DWF, gross2006:hudson}.

\medskip
The ressemblance between contextuality and negativity was first noticed by Spekkens who generalized
these two notions in order to prove that they coincide \cite{spekkens2008:equivalence}.
However, this result remains difficult to apply, since a large number of Wigner functions must be 
probed to identify contextuality. 
Howard {\em et al.} \cite{howard2014:contextuality} showed that, 
if one restricts to a particular class of measurements, namely Pauli measurements,
one can select a particular Wigner function which allows by itself to characterize contextuality. 
They proved that contextuality for stabilizer measurements and negativity of Gross' 
Wigner function \cite{gross2006:hudson} are equivalent for a single qudit of odd prime dimension. 
In the present work, we extends this result to multiple qudits and to any odd local dimension.
We prove that, for quantum systems of odd local dimension, these two notions are equivalent.
Such a neat equivalence between these features introduced in different fields is quite 
unexpected. Indeed, while contextuality is a concept from the fondation of quantum physics, 
Wigner functions originate from quantum optics.
In addition, our proof of this equivalence is much more straightforward. 
Indeed, we directly compute the value of the Wigner function in terms of the hidden 
variable model and we observe that if the hidden variable model is non-contextual then 
the Wigner function is non-negative.

\medskip
Our proof of this equivalence relies on the choice of a simple definition of 
contextuality based on value assignments whereas the work of Howard {\em et al.} 
is based on the graphical formalism of Cabello {\em et al.} \cite{cabello2014:graphcontextuality}.
The relations between these different notions of non-classicality are depicted in Fig.~\ref{fig:NC_equivalences}.
In this work, we consider only contextuality of stabilizer measurements and hidden variable models
are assumed to be deterministic.

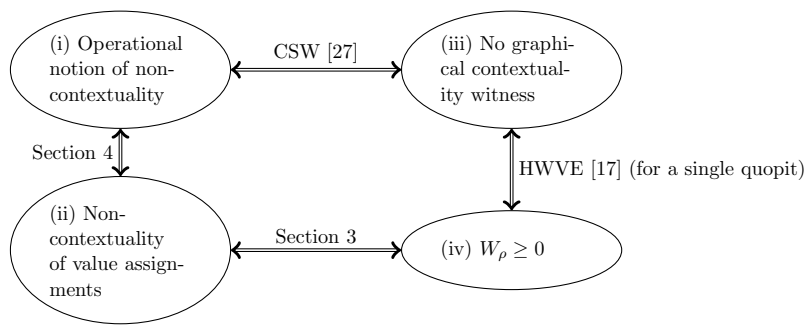
\begin{figure}[ht]
\captionsetup{margin=1.5cm}
\centering
\begin{tikzpicture}[scale=.8, every node/.style={scale=.7}]
\def\r{4}

\draw[ellipse, minimum size = 1.5cm, text width=2.7cm]
	(0,1.5) node[draw] (i) {(i) Operational notion of non-contextuality}
	(0,-1.5) node[draw] (ii) {(ii) Non-contextuality of value assignments}
	(6.5,1.5) node[draw] (iii) {(iii) No graphical contextuality witness}
	(6.5,-1.5) node[draw] (iv) {(iv) $W_\rho \geq 0$};

\draw[<->, double]
	(i) -- (ii) node [midway, left]{Section~\ref{sec:equiv_NC_def}};
\draw[<->, double]
	(i) -- (iii) node [midway, above]{CSW \cite{cabello2014:graphcontextuality}};
\draw[<->, double]
	(iii) -- (iv) node [midway, right]{HWVE \cite{howard2014:contextuality} (for a single quopit)};
\draw[<->, double]
	(iv) -- (ii) node [midway, above]{Section~\ref{sec:equiv_NC_negativity}};

\end{tikzpicture}
\caption{
Relation between different notions of non-classicality.
}
\label{fig:NC_equivalences}
\end{figure}

%
%
%

\medskip
This article is organized as follows.
The necessary stabilizer formalism is recalled in Section~\ref{sec:stabilizer}.
Section~\ref{sec:equiv_NC_negativity} introduces a notion contextuality 
based on value assignment and prove the equivalence between this notion
and the negativity of Gross' discrete Wigner function, {\em i.e.} 
(ii)~$\Leftrightarrow$~(iv).
The purpose of Section~\ref{sec:equiv_NC_def} is to prove the equivalence 
between the 2 notions of contextuality (i) and (ii), completing the square 
in Fig.~\ref{fig:NC_equivalences}.

\section{Background on the stabilizer formalism} \label{sec:stabilizer}

In what follows, we consider the Hilbert space $\h = (\C^d)^{\otimes n}$,
where $d$ is an odd integer and $n$ is a non-negative integer.
We consider an orthonormal basis $(\ket a)_{a \in \Z_d}$ of $\C^d$.
The $n$-fold tensor products of these vectors provide an orthonormal
basis  $(\ket{ {\bf a}})_{ {\bf a} \in \Z_d^n}$ of the Hilbert space $\h$ indexed by $\Z_d^n$.

\medskip
The space $V = \Z_d^{n} \times \Z_d^{n}$ that is called the {\em phase space}
and will be used to index Pauli operators acting on $\h$.
Vectors in $V$ are denoted $({\bf u_Z},{\bf u_X})$, 
where both ${\bf u_Z}$ and ${\bf u_X}$ live in $\Z_d^n$.
The space $\Z_d^n$ is equipped with the standard inner product
$({\bf a}|{\bf b}) = \sum_{i=1}^n a_i b_i$ for all ${\bf a, b} \in \Z_d^n$,
whereas the phase space $V$ is equipped with the symplectic inner product 
defined by 
$$
[{\bf u,v}] = ({\bf u_Z}|{\bf v_X}) - ({\bf u_X}|{\bf v_Z}) \mod d
$$
where 
${\bf u} = ({\bf u_Z}, {\bf u_X}) \in V$ and 
${\bf v} = ({\bf v_Z}, {\bf v_X}) \in V$.

\medskip
{\em Pauli matrices} can be generalized to obtain matrices acting on $\C^d$
as follows. 
Let $\omega$ be the $d$-th root of unity, $\omega = e^{2i\pi/d}$.
The Pauli matrices $X$ and $Z$ are defined by
$$
X \ket a = \ket{a+1}, \quad Z \ket a = \omega^a \ket a
$$
for all $a \in \Z_d$.
Tensor products of these matrices are denoted 
$Z^{{\bf a}} = Z^{a_1} \otimes \dots Z^{a_n}$ and
$X^{{\bf a}} = X^{a_1} \otimes \dots X^{a_n}$
where ${\bf a} = (a_1, \dots, a_n) \in \Z_d^n$. They 
satisfy $Z^{{\bf a}} X^{{\bf b}} = \omega^{({\bf a} | {\bf b})} X^{{\bf b}} Z^{{\bf a}}$.

\medskip
{\em Generalized Pauli operators} acting on this Hilbert space $\h$ are of the form 
$\omega^{a} Z^{{\bf u_Z}} X^{{\bf u_X}}$ where ${\bf u = (u_Z, u_X)} \in V$, 
and $a \in \Z_d$.
Just as qubits Pauli operators, they form a group.
We fix the phase of $Z^{{\bf u_Z}} X^{{\bf u_X}}$ 
to be $\omega^{-({\bf u_Z|u_X})/2}$. This defines {\em Heisenberg-Weyl operators} 
\begin{equation} \label{eqn:Tu}
T_{\bf u} = \omega^{-({\bf u_Z|u_X})/2} Z^{{\bf u_Z}} X^{{\bf u_X}}
\end{equation}
for all ${\bf u} \in V$. Recall that, since $d$ is an odd integer, the term 
$({\bf u_Z|u_X})/2$ in Eq.\eqref{eqn:Tu} is a well defined element of $\Z_d$.
This choice is well suited to describe measurement outcomes since $T_{\bf u} T_{\bf v} = \omega^{[{\bf u, v}]/2} T_{{\bf u+v}}$,
in particular two operators $T_{\bf u}$ and $T_{\bf v}$ can be measured simultaneously if and only
if $[{\bf u, v}] = 0$, which implies $T_{\bf u} T_{\bf v} = T_{{\bf u+v}}$.
Their commutation relation depends on the symplectic inner product as follows
$T_{\bf u} T_{\bf v} = \omega^{[{\bf u, v}]}T_{\bf v} T_{\bf u}$.
These operators satisfy $(T_{\bf u})^d = 1$ which proves that their eigenvalues belong to the group
$U_d = \{ \omega^{s} \ | \ {s} \in \Z_d \}$ of $d$-th roots of unity.

\medskip
Measuring a family of $m$ mutually commuting operators 
$C = \{T_{{\bf a_1}}, \dots, T_{{\bf a_m}} \}$
returns the outcome ${\bf s} = (s_1, \dots, s_m) \in \Z_d^m$ with probability
$\Tr(\Pi_{\bf u}^{\bf s})$, where $\Pi_{\bf u}^{\bf s}$ 
is the projector onto the common eigenspace of the operators $T_{{\bf a_i}}$
with respective eigenvalue $\omega^{s_i}$.
When no confusion is possible, this projector is simply denoted $\Pi_C^{\bf s}$.
For a single operator $T_{\bf u}$, we have 
$T_{\bf u} = \sum_{{\bf s} \in \Z_d} \omega^{\bf s} \Pi_{\bf u}^{\bf s}$,
by definition of these projectors.
Moreover, $\Pi_{\bf u}^{\bf s}$ can be obtained from the operator
$T_{\bf u}$ as
\begin{equation} \label{eqn:projector_single_meas}
\Pi_{\bf u}^{\bf s} = \frac{1}{d} \sum_{k \in \Z_d} \omega^{-k{\bf s}} T_{\bf u}^k \cdot
\end{equation}
In order to generalize Eq.\eqref{eqn:projector_single_meas} 
to a family $C = \{T_{{\bf a_1}}, \dots, T_{{\bf a_m}} \}$ of operators, 
introduce the $\Z_d$-linear subspace 
$M_{\bf a} = \langle {\bf a_1, \dots, a_m} \rangle$ of $V$
generated by the vectors ${\bf a_i}$.
Any measurement outcome ${\bf s} \in \Z_d^m$ induces a $\Z_d$-linear 
form $\ell_{\bf s}: M_{\bf a} \mapsto \Z_d$ defined by
$
\ell_{\bf s}( \sum_{i=1}^m x_i {\bf a_i}) = \sum_{i=1}^m x_i s_i
$
where $x_i \in \Z_d$ for all $i=1, \dots, m$.
This map $\ell_{\bf s}$ parametrizes the group generated 
by the operators $\omega^{s_i} T_{\bf a_i}$ as follows
$
\langle \omega^{s_i} T_{\bf a_i} \rangle = \{ \omega^{\ell_{\bf s}({\bf u})} T_{\bf u} \ | \ {\bf u} \in M_{\bf a} \}
$.
The projector $\Pi_{\bf a}^{\bf s}$ can be written
\begin{equation} \label{eqn:projectors_sum}
\Pi_{\bf a}^{\bf s} = \frac{1}{|M_{\bf a}|} \sum_{{\bf u} \in M_{\bf a}}  \omega^{\ell_{\bf s}({\bf u})} T_{\bf u} \cdot
\end{equation}

\section{Contextuality of value assignments and negativity} \label{sec:equiv_NC_negativity}

In this section, we prove that the notion of non-contextuality based on the existence
of coherent value assignments is equivalent to the non-negativity of the discrete Wigner 
function. This equivalence was proven by Howard {\em et al.} \cite{howard2014:contextuality} 
in the case of a single qudit in odd prime dimension. 
We propose a simple proof of this result which allows us to generalize this equivalence 
to multiple qudits and to any odd local dimension.

\medskip
Recall that contextuality refers to the fact that measurement outcomes cannot be described in a deterministic 
way. One cannot associate a fixed outcome $\lambda(A)$ with each observable $A$ in such a way 
that this value is simply revealed after measurement. The algebraic relations between 
compatible observables must be satisfied by outcomes as well, making the existence of 
such pre-existing outcomes impossible.
For instance, given two commuting operators $A$ and $B$, the three observables
$A$, $B$ and $C=AB$ can be measured simultaneously and the values $\lambda(A), \lambda(B)$ 
and $\lambda(C)$ associated with these operators must satisfy 
$\lambda(C) = \lambda(A) \lambda(B)$. No value assignment satisfying all these algebraic 
constraints exists in general.

\medskip
The Wigner functions of a state is a description of this state that was 
introduced in quantum optics in order to identify states with a classical 
behavior. Quantum states with a non-negative Wigner function are 
considered as quasi-classical states. The non-negativity of the Wigner function
allows to describe the statistics of the outcomes of a large class of measurements
in a classical way.
The success of this representation motivated their generalization to finite
dimensional Hilbert spaces, called discrete Wigner function. 
Different generalizations have been considered and 
finding a finite-dimensional Wigner representation that behaves as nicely as its original 
infinite-dimensional version \cite{wigner1932:WF} of use in quantum optics is a non-trivial question.
The qubit case illustrates the difficulty of this task \cite{delfosse2015:rebits, raussendorf2015:contextuality_QCSI}.
In this work, we restrict ourselve to system of qudits with odd local dimension and we 
consider Gross' discrete Wigner function which encloses most of the features of
its quantum optics counterpart \cite{gross2006:hudson}.

\subsection{Value assignments are characters in odd local dimension}

In this definition, the HVM associates a deterministic eigenvalue
$\lambda_{\nu}({\bf a}) \in U_d$ with each operator $T_{\bf a}$. 
Measurements only reveal this pre-existing values 
that are independent on other compatible measurements being performed.
Formally, non-contextual value assignments are defined as 
follows.

\begin{defi} \label{def:value_assignment}
Let $\rho$ be a density matrix over $\h$.
A {\em set of NC value assignments} for the state $\rho$ is a triple $(S, q_\rho, \lambda)$ where 
$S$ is a finite set, $q_\rho$ is a probability distribution over $S$ and $\lambda$ is
a collection of maps $\lambda_\nu: V \rightarrow U_d$, for each $\nu \in S$, such that
\begin{enumerate}
\item for all ${\bf u, v} \in V, [{\bf u, v}] = 0$ implies $\lambda_\nu({\bf u + v}) = \lambda_\nu({\bf u}) \lambda_\nu({\bf v})$,
\item for all ${\bf u} \in V$,
$
\Tr(T_{\bf u} \rho) = \sum_{\nu \in S} \lambda_\nu({\bf u}) q_{\rho}(\nu).
$
\end{enumerate}
\end{defi}

Without loss of generality, one can assume that the value assignment associated
with distinct states $\nu \in S$ are distincts. 
Then, $S$ is necessarily a finite set since there is only a finite number of distinct 
maps $\lambda_\nu$ from $V$ to $U_p$.
A map $\lambda_\nu$ which satisfies 
$\lambda_\nu({\bf u + v}) = \lambda_\nu({\bf u}) \lambda_\nu({\bf v})$
for all ${\bf u, v}$ such that $[{\bf u, v}] = 0$, is called a {\em value assignment}.
Recall that the value $\lambda_\nu({\bf u})$ is associated with the operator $T_{\bf u}$
defined in Eq.\eqref{eqn:Tu}.
With this phase convention, a value assignment is defined in such a 
way that the value of the product $T_{\bf u+v} = T_{\bf u} T_{\bf v}$ 
of two compatible operators is the products of their values.
As we will see later, multiplicativity of $\lambda_\nu$ whenever $[{\bf u, v}]=0$ is the
non-contextuality assumption whereas the condition
$\Tr(T_{\bf u} \rho) = \sum_{\nu \in S} \lambda_\nu({\bf u}) q_{\rho}(\nu)$
ensures that this triple is sufficient to recover the prediction of quantum mechanics.

\medskip
A value assignment is a map $\lambda_\nu$ that satisfies the 
constraint $\lambda_\nu({\bf u + v}) = \lambda_\nu({\bf u}) \lambda_\nu({\bf v})$,
for all pairs of vectors such that $[{\bf u, v}] = 0$.
For instance, the $|V| = d^{2n}$ characters of $V$, which satisfy this constraint for all pairs 
${\bf u, v}$, provide $d^{2n}$ value assignments $\lambda_\nu$.
However, if value assignments $\lambda_\nu: V \rightarrow \C^*$ coicinde with a character
of $V$ over any isotropic subspace, nothing in Def.\ref{def:value_assignment} guarantees
that these assignments are actually characters of $V$. The next lemma proves this property,
\emph{i.e.} the only consistent value assignments of the HVM are given by the $d^{2n}$ characters of $V$.

\begin{lemma} \label{lemma:value_are_characters}
For any odd integer $d > 1$ and for any integer $n \geq 2$,
value assignments $\lambda_\nu$ are characters of $V$.
\end{lemma}

\begin{proof}
To make the proof easier to follow, we regard $\lambda$ as a map
from $V$ to $\Z_d$ (through the group isomorphism between $U_d$ and $\Z_d$) 
and we use the additive notation $\lambda({\bf u}) + \lambda({\bf v})$
instead of $\lambda({\bf u}) \lambda({\bf v})$ in $\C$.
We already know that $\lambda({\bf u+v}) = \lambda({\bf u}) + \lambda({\bf v})$
whenever $[{\bf u, v}] = 0$.
In particular for all ${\bf u} \in V$ and for all $k \in \Z_d$ we have 
$\lambda(k{\bf u}) = k\lambda({\bf u})$.

Consider the canonical basis of $\Z_d^n \times \Z_d^n$ that we denote 
$({\bf e_1, \dots, e_n, f_1, \dots f_n})$. Clearly, $[{\bf e_i, f_i}] = 1$ and the 
planes $P_i = \langle {\bf e_i, f_i} \rangle$ are pairwise orthogonal
for all $i = 1, \dots, n$.
The orthogonality between these planes allows us to write
$$
\lambda({\bf u}) = \sum_{i=1}^n \lambda(\alpha_i {\bf e_i} + \beta_i {\bf f_i})
$$
for all ${\bf u} = \sum_{i=1}^n (\alpha_i {\bf e_i} + \beta_i {\bf f_i})$.
It remains to prove that the restriction of $\lambda$ to any plane $P_i$ is additive.

We will use a second plane $P_j = \langle {\bf e_j, f_j} \rangle$ with 
$j \neq i$ (which exists only when $n\geq 2$).
Denote ${\bf u} = \alpha_i {\bf e_i}$ and ${\bf v} = \beta_i {\bf f_i}$ and define 
${\bf u'} = \beta_i {\bf e_j}$ and ${\bf v'} = \alpha_i {\bf f_j}$ so that $[{\bf u, v}] = [{\bf u', v'}]$.
To conclude it suffices to prove that
$\lambda({\bf u+v}) = \lambda({\bf u}) + \lambda({\bf v})$.
Write
$$
{\bf u+v} = \frac{1}{2}\left( ({\bf u + v + u' + v'}) + ({\bf u + v - u' - v'}) \right) \cdot
$$
This decomposition is chosen in such a way that $({\bf u + v + u' + v'})$ and
$ ({\bf u + v - u' - v'})$ are orthogonal:
$$
{\bf [(u + v + u' + v'), (u + v - u' - v')] = [u, v] + [v, u] + [u', -v'] + [v', -u']} = 0\cdot
$$
We will also use the orthogonality relations $[{\bf u\pm v', v \pm u'}] = 0$
and between the planes $P_i$ and $P_j$.
We obtain
\begin{align*}
\lambda({\bf u+v}) 
& = \lambda \left( \frac{1}{2} ({\bf u + v + u' + v'}) +  \frac{1}{2}({\bf u + v - u' - v'})  \right)\\
& =  \frac{1}{2} \lambda ({\bf u + v + u' + v'}) +  \frac{1}{2}\lambda({\bf u + v - u' - v'})\\
& =  \frac{1}{2} \lambda\left({\bf (u + v') + (v + u')}\right) +  \frac{1}{2}\lambda\left({\bf (u - v') + (v - u')}\right)\\
& =  \frac{1}{2} \left( \lambda({\bf u + v'}) + \lambda({\bf v + u'}) +  \lambda({\bf u - v'}) +\lambda({\bf v - u'}) \right)\\
& =  \frac{1}{2} \left( \lambda({\bf u}) + \lambda({\bf v'}) + \lambda({\bf v}) + \lambda({\bf u'}) +  \lambda({\bf u}) - \lambda({\bf v'}) +\lambda({\bf v}) - \lambda({\bf u'}) \right))\\
& =  \lambda({\bf u}) + \lambda({\bf v})\cdot
\end{align*}
This proves that $\lambda$ is a character.
\end{proof}

\subsection{Discrete Wigner functions}

The purpose of this section is to recall that a set of NC value assignments can be 
derived from a discrete Wigner function of a state whenever this function is non-negative.

\medskip
Quantum states are generally represented by their density matrix $\rho$.
The {\em Wigner function} $W_\rho$ of a state $\rho$ is an alternative 
description of the state $\rho$. This representation is sometimes more convenient 
than the density matrix.
Wigner functions have been introduced in quantum optics \cite{wigner1932:WF} 
and the negativity of the Wigner function of a state is regarded as an indicator of 
non-classicality of this quantum state.
In the present work, we consider their finite dimensional generalization which is 
called {\em discrete Wigner function} \cite{gibbons2004:DWF, galvao2005:DWF_speedup}.

\medskip
We focus on Gross' discrete Wigner function \cite{gross2006:hudson} $W_\rho: V \rightarrow \R$ 
defined by 
$
W_\rho({\bf u}) = d^{-n} \Tr(A_{\bf u} \rho)
$
where $A_{\bf u} =  d^{-n} \sum_{{\bf v} \in V} \omega^{[{\bf u, v}]}T_{\bf v}$.
The family $(A_{\bf u})_{{\bf u} \in V}$ is an orthonormal
basis of the space of $(d^n \times d^n)$-matrices
equipped with the inner product $(A | B) = d^{-n} \Tr(A^\dagger B)$.
The values $W_\rho({\bf u})$ for ${\bf u} \in V$ are simply the coefficients of 
the decomposition of the matrix $\rho$ in this basis:
\begin{equation} \label{eq:W_rho_decomposition}
\rho = \sum_{{\bf u} \in V} W_\rho({\bf u}) A_{\bf u}.
\end{equation}
This proves that the Wigner function $W_\rho$ fully describes 
the state $\rho$.

\medskip
In order to describe measurement outcomes, the Wigner representation 
can be extended to POVM elements $(E_s)_s$.
The {\em Wigner function of a POVM element} $E_s$ is defined to be
$
W_{E_s}({\bf u}) = \Tr(E_s A_{\bf u}).
$
This definition is chosen in such a way that the probability $\Tr(E_s \rho)$ of the 
outcome $s$ is given in terms of the Wigner functions $W_\rho$ and $W_{E_s}$
by
\begin{equation} \label{eqn:proba_Wigner}
\Tr(E_s \rho) = \sum_{{\bf u} \in V}W_{E_s}({\bf u}) W_\rho({\bf u})
\end{equation}
This expression is obtained by replacing $\rho$ by its decomposition \eqref{eq:W_rho_decomposition}.

\medskip
Consider for instance the measurement of an operator $T_{\bf a}$.
It corresponds to the POVM elements $(\Pi_{\bf a}^s)_{s \in \Z_d}$.
Calculating the value of the Wigner function of the POVM element $\Pi_{\bf a}^s$,
we find
$
W_{\Pi_{\bf a}^s}({\bf u}) = \delta_{[{\bf a,u}], s},
$
and injecting this in Eq.\eqref{eqn:proba_Wigner} yields
$
\Tr(\Pi_{\bf a}^s \rho) = \sum_{{\bf u} \in V} \delta_{[{\bf a,u}], s}.
$
Using the decomposition of $T_{\bf a}$ as a sum of projectors, this provides the 
expectation of $T_{\bf a}$.
\begin{lemma} \label{lemma:expectation_Wigner}
For all Heisenberg-Weyl operators $T_{\bf a}$ given in Eq.\eqref{eqn:Tu}, it holds that
$$
\Tr(T_{\bf a} \rho) = \sum_{{\bf u} \in V} W_{\rho}({\bf u}) \omega^{[{\bf a, u}]} \cdot
$$
\end{lemma}

Comparing Lemma~\ref{lemma:expectation_Wigner} with Def.~\ref{def:value_assignment},
we see that the triple $(V, W_\rho, (\lambda_{\bf u})_{{\bf u} \in V} )$ is a set of NC value assignments for
$\lambda_{\bf u}({\bf a}) = \omega^{[{\bf a, u}]}$ whenever the Wigner function of the
state $\rho$ is non-negative. Note that all characters of $V$ can be written as
$\omega^{[*, {\bf u}]}$ for some ${\bf u} \in V$.

\subsection{Equivalence between NC value assignments and non-negativity}

The previous section shows that non-negativity of $W_\rho$ implies the existence of 
a set of NC value assignments for the state $\rho$. We now prove the converse statement.

\begin{prop} \label{prop:NC_implies_non_negativity}
Let $n \geq 2$ and let $\rho$ be a $n$-qudits state of odd local dimension $d > 1$.
If $\rho$ admits a set of NC value assignments 
then $W_\rho \geq 0$.
\end{prop}

\begin{proof}
Given a set of NC value assignments, let us compute the value of the Wigner 
function of $\rho$ to prove that it is non-negative. We have
\begin{align*}
W_{\rho}({\bf u}) 
& = d^{-n} \Tr(A_{\bf u} \rho)\\
& = d^{-2n} \sum_{{\bf v} \in V} \omega^{[{\bf u, v}]} \Tr( T_{\bf v} \rho)\\
& = d^{-2n} \sum_{{\bf v} \in V} \omega^{[{\bf u, v}]} \sum_{\nu \in S} \lambda_{\nu}({\bf v}) q_{\rho}(\nu)\\
& = d^{-2n} \sum_{\nu \in S} \left( \sum_{{\bf v} \in V}  \omega^{[{\bf u, v}]} \lambda_{\nu}({\bf v}) \right) q_{\rho}(\nu)
\end{align*}
The fact that such a sum is positive or even real is not clear. However, Lemma~\ref{lemma:value_are_characters} shows that $\omega^{[{\bf u}, \cdot]} \lambda_{\nu}(\cdot)$ is 
also a character. Hence, any sum $\sum_{{\bf v} \in V} \omega^{[{\bf u, v}]} \lambda_{\nu}({\bf v})$ is either 0 or $d^{2n}$ \cite{lang2002:algebra}.
Since $q_\rho(\nu) \geq 0$, this proves that $W_\rho({\bf u}) \geq 0$.
\end{proof}

Actually the HVM derived from the discrete Wigner function is essentially unique.
The following corollary shows that the distribution $q_\rho$ of any set of NC
value assignments can be identified with the Wigner function distribution $W_\rho$
over $V$.

\begin{coro}
Let $n \geq 2$ and let $\rho$ be a $n$-qudits state of odd local dimension $d > 1$.
If $(S, q_\rho, \lambda)$ is a NC set of value assignments for $\rho$
then there exists an bijective map $\sigma: S \rightarrow V$ such that
$$
q_{\rho}(\nu) = W_\rho(\sigma(\nu))
$$
for all $\nu \in S$.
\end{coro}

\begin{proof}
Let us refine the argument of the proof of Proposition~\ref{prop:NC_implies_non_negativity}.
By Lemma~\ref{lemma:value_are_characters} and since value assignments 
corresponding to distinct states of $S$ are assumed to be different,
there exists an injective map $\phi: S \rightarrow V$ such that 
$\lambda_\nu = \omega^{[\phi(\nu), \cdot]}$.
Without loss of generality, we can assume that $\phi$ is sujective by adding
extra elements to $S$ corresponding to the missing characters of $V$. 
For these new elements $\nu \in S$, we simply set $q_\rho(\nu) = 0$ to 
preserve the prediction of this triple.
With this notation, the expression of $W_\rho({\bf u})$ obtained in the previous 
proof becomes
\begin{equation} \label{eqn:proof:VA_to_W}
W_\rho({\bf u})
= d^{-2n} \sum_{\nu \in S} \left( \sum_{{\bf v} \in V}  \omega^{[{\bf u} + \phi(\nu), {\bf v}]} \right) q_{\rho}(\nu)
= \sum_{\nu \in S} \delta_{{\bf u} + \phi(\nu), {\bf 0}} \cdot q_{\rho}(\nu) \cdot
\end{equation}
For all vectors $\bf u$, there exists a unique state $\nu \in S$ such that 
${\bf u} + \phi(\nu) = {\bf 0}$. Denote by $\nu_{\bf u}$
$\phi'({\bf u})$ 
this state. Then, Eq.\eqref{eqn:proof:VA_to_W} becomes
$
W_\rho({\bf u}) = q_{\rho}(\nu_{\bf u}).
$
To conclude the proof, note that the map ${\bf u} \rightarrow \nu_{\bf u}$ is invertible,
by bijectivity of $\phi$.
Its inverse is the map $\sigma$ of the corollary.
\end{proof}

\section{Operational definition of non-contextuality} \label{sec:equiv_NC_def}

Our proof of the equivalence between non-contextuality and non-negativity
of the Wigner function relies on the choice of a simple definition of contextuality
(Def.~\ref{def:value_assignment}).
The purpose of this section is to prove that, although simple, this definition is sufficient 
to capture the same notion of contextuality as the one considered in previous
work.
Namely, we prove that the notion of NC contextual value assignments is equivalent to
the operational definition of contextuality given below.

\subsection{Deterministic hidden variable model}

The role of a hidden variable model (HVM) is to describe the outcome distribution of 
Pauli measurements for a state $\rho$.
In what follows, an ordered family of commuting operators 
$C = \{ T_{{\bf a_1}}, \dots, T_{{\bf a_m}} \}$ 
is called a {\em context}.
This guarantees that they can be measured simultaneously.
We denote by $\cal C$ the set of all the contexts, that is of all ordered tuples of commuting 
Heisenberg-Weyl operators $T_u$.

\medskip
The HVM considered in Def.~\ref{def:value_assignment} is designed to predict 
the expectation of single Pauli operators. We now consider a framework which is 
{\em a priori} more general. 
The HVM is required to predict the outcome distribution of the measurement of
any Pauli context $C \in {\cal C}$.

\begin{defi} \label{defi:HVM}
A HVM for a state $\rho$ is defined to be a triple $(S, q_\rho, p)$ 
such that $S$ is a set, $q_\rho$ is a distribution over $S$ and
$p = (p_C)_{C \in {\cal C}}$ is a family of conditional probability distributions
$p_C({\bf s} | \nu)$ over the possible outcome ${\bf s}$ for any context $C \in {\cal C}$ 
and for any state $\nu \in S$.
We require that the prediction of the HVM matches the quantum mechanical 
value, \emph{i.e.}
\begin{align} \label{eq:HVM_prediction}
\Tr(\Pi_{C}^{\bf s} \rho) = \sum_{\nu \in S} p_C({\bf s} | \nu) q_\rho(\nu),
\end{align}
for all context $C \in {\cal C}$ and for all outcome ${\bf s} \in \Z_d^{m}$.
\end{defi}

The set $S$ is the set of states of the HVM. Note that this set is not {\em a priori} 
the same as the set used in Def.~\ref{def:value_assignment}.
We understand the probability $q_\rho(\nu)$ as the probability that the system is in the 
state $\nu \in S$.
As suggested by the notation, the probability $p_C(s | \nu)$ can be interpreted as the 
probability of an outcome ${\bf s}$ when measuring the operators of the context $C$ and 
when the system is in the state $\nu$ of the HVM. It is then natural to define the 
prediction of the HVM as in Eq.\eqref{eq:HVM_prediction}.

\medskip
In the present work, we consider HVM that are deterministic in the following sense. We
can associate a fixed measurement outcome $\bf s$ with each state $\nu$ of the model.
In other words, conditional probabilities $p_C({\bf s} | \nu)$ are delta functions.

\begin{defi} \label{def:HVM_deterministic}
A HVM $(S, q_\rho, p)$ is said to be {\em deterministic} if conditional distributions $p_C$
are all delta functions, {\em  i.e.} for all $C \in {\cal C}$ and for all outcome ${\bf s}$, we have
$
p_{C}({\bf s} | \nu) = \delta_{\alpha_\nu(C), {\bf s}}
$
for some map $\alpha_\nu$ that associates a value $\alpha_\nu(C) \in \Z_d^{|C|}$
with each context.
\end{defi}

We often denote such a HVM by the triple $(S, q, \alpha)$. 
When $C$ contains $m$ operators, both $\alpha_\nu(C)$ and ${\bf s}$ are
$m$-tuples. If $\alpha_\nu(C)  = (x_1, \dots, x_m)$, then 
$\delta_{\alpha_\nu(C), {\bf s}}$ is the product of the delta functions 
$\delta_{x_i, s_i}$.

\subsection{Operational definition of non-contextuality}

Within our formalism restricted to measurement of Pauli operators $T_{\bf a}$,
there exists different ways to realize a measurement. The operational notion
of contextuality refers to the fact that the conditional distribution of outcomes in 
the HVM may depend on the way the measurement is {\em implemented}.
This section presents a formal definition of this notion.

\medskip
To illustrate what should be the right definition of an implementation, we start 
with some examples.
We can measure the operators of a context 
$C = \{T_{\bf a_1}, \dots, T_{\bf a_\ell} \} \in {\cal C}$
and return the corresponding outcome ${\bf s} = (s_1, \dots, s_\ell)$.
This realizes the measurement defined by the family of orthogonal projectors
$(\Pi_{\bf a}^{\bf s})_{\bf s}$ for ${\bf s} \in \Z_d^m$.
Different contexts $C$ may produce the same family of projectors, that is
the same measurement.
For instance, the 2-qudit measurement defined by the projectors
$
P^{x_1, x_2} = \proj {x_1} \otimes \proj {x_2}, 
$
indexed by $(x_1, x_2) \in \Z_d^2$,
can be implemented via the contexts $C = \{X \otimes I, I \otimes X\}$
or alternatively via $C' = \{X \otimes I, X \otimes X\}$.

\medskip
We can also measure of a family of operators but read only a subset of the outcomes 
or even of function of the outcomes. 
Such a classical postprocessing extends the set of projectors that can be reached.
For instance, for ${\bf u} \in V$ consider a pair of projectors 
$\{ \Pi^0 = \Pi_{\bf u}^0, \Pi^1 = I - \Pi_{\bf u}^0 \}$.
This measurement is realized by measuring $T_{\bf u}$ and returning $0$ 
if the outcome is 0 and returning 1 for all other outcome $s \in \Z_d \backslash \{0\}$.
To provide an example with a less trivial postprocessing consider 
the measurement $(\Pi_{\bf u+v}^t)$ with outcome $t \in \Z_d$, for ${\bf u, v} \in V$
such that $[{\bf u, v}] = 0$.
We can realize this measurement by measuring the pair $C = \{T_{\bf u}, T_{\bf v}\}$ and 
by returning only the sum $t = s_{\bf u} + s_{\bf v} \in \Z_d$ of the two outcomes.

\medskip
These examples motivate our definition of an implementation of a measurement.
Consider a measurement defined by a family of stabilizer projectors $(\Pi^s)_{s \in O}$,
that sum up to identity, indexed by the elements of a finite set $O$.
\begin{defi} \label{defi:implementation}
An {\em implementation} of a measurement $(\Pi^{\bf s})_{{\bf s} \in O}$, is defined to be
a pair $(C, \sigma_C)$ where $C = \{T_{\bf a_1}, \dots, T_{\bf a_\ell} \} \in {\cal C}$ 
and $\sigma_C : \Z_d^\ell \rightarrow O$ is a surjective postprocessing map such that
$$
\Pi^{\bf s} = \sum_{\substack{{\bf t} \in \Z_d^\ell\\ \sigma_C({\bf t}) = {\bf s}}} \Pi_{\bf a}^{{\bf t}}
$$
for all ${\bf s} \in O$.
\end{defi}
Neither the choice of the context $C$ nor the corresponding map $\sigma_C$
is unique. The postprocessing map is assumed to be surjective only to ensure that 
all the projectors of the family $(\Pi^{\bf s})_{{\bf s} \in O}$ are reached.

\medskip
Let $(S, q_\rho, \nu)$ be a HVM describing measurement outcomes 
for a state $\rho$.
Consider an implementation $(C, \sigma_C)$ of a measurement $(\Pi^{\bf s})_{{\bf s} \in O}$.
By definition of the projectors $\Pi^{\bf s}$, the HVM predicts that the outcome ${\bf s}$ 
occurs with probability
$$
p_{C}(\sigma_C^{-1}({\bf s}) | \nu) = \sum_{\substack{{\bf t} \in \Z_d^{|C|}\\ \sigma_C({\bf t}) = {\bf s}}}  p_C( {\bf t} | \nu)
$$
when the system is in position $\nu$.
Quantum mechanics predicts that the distribution of the outcome of a 
measurement only depends on the projectors $\Pi^{\bf s}$ and not on the 
implementation. That means that for any two implementations $(C, \sigma_C)$
and $(C', \sigma_{C'})$ of a measurement $\Pi^{\bf s}$, we have
\begin{align*}
\Tr(\Pi^{\bf s} \rho) 
& = \sum_{\nu \in S} p_{C}(\sigma_C^{-1}({\bf s}) | \nu) q_\rho(\nu)\\
& = \sum_{\nu \in S} p_{C'}(\sigma_{C'}^{-1}({\bf s}) | \nu) q_\rho(\nu)
\end{align*}
However, nothing in Def.~\ref{defi:HVM} requires that the probabilities
$p_{C}(\sigma_{C}^{-1}({\bf s}) | \nu)$ and $p_{C'}(\sigma_{C'}^{-1}({\bf s}) | \nu)$ coincide
for all $\nu \in S$.
This assumption is a notion of non-contextuality.

\begin{defi} \label{defi:NC_operational}
A HVM $(S, q_\rho, p)$ is said to be {\em non-contextual} if
for all implementation $(C, \sigma_C)$ of a measurement $(\Pi^{\bf s})_{{\bf s} \in O}$,
for all $\nu \in S$, the conditional probability
$
p_C(\sigma_C^{-1}({\bf s}) | \nu)
$
depends only on the projector $\Pi^{\bf s}$ and not on the 
implementation $(C, \sigma_C)$.
\end{defi}

For instance, we saw that the measurement $(\Pi^{s}_{\bf u+v})_{{s} \in \Z_d}$ can be 
implemented by $C = \{T_{\bf u+v}\}$ with a trivial map $\sigma_C$ but also
using $C' = \{T_{\bf u}, T_{\bf v}\}$ with the postprocessing map $\sigma_{C'}(s_{\bf u}, s_{\bf v}) = s_{\bf u}+s_{\bf v}$.
For a non-contextual model, the corresponding conditional probabilities
$p_C({s} | \nu)$ and $p_{C'}( \sigma_{C'}^{-1}({s}) | \nu)$ coincide for all states
$\nu$ of the HVM. In this example, we have 
$\sigma_{C'}^{-1}({s}) = \{ ({t}, {s-t}) \ | \ {t} \in \Z_d \}$.

\subsection{Equivalence of the two definitions of non-contextuality}
 \label{section:equiv_NCvalue_NCHVM}

In this section we prove that the existence of a NCHMV as given in Def.~\ref{defi:NC_operational}
and the existence of a set of NC value assignments as in Def.~\ref{def:value_assignment} 
are two equivalent notions of non-contextuality. 
This is the equivalence (i) $ \Leftrightarrow $ (ii) in Fig.~\ref{fig:NC_equivalences}.

\medskip
The following proposition shows that the implication (ii) $\Rightarrow$ (i) holds.

\begin{prop} \label{prop:NCvalue_implies_NCHVM}
Let $(S, q_\rho, \lambda)$ be a set of NC value assignments.
Then $(S, q_\rho, p)$ defines a deterministic NCHVM 
where $p$ is defined by
\begin{align} \label{eqn:pC_from_values}
p_C({\bf s} | \nu) = \frac{1}{|M_{\bf a}|} \sum_{{\bf u} \in M_{\bf a}} \omega^{\ell_s({\bf u})} \cdot \lambda_\nu(u)
\end{align}
for all $C = \{T_{\bf a_1}, \dots, T_{\bf a_m}\} \in {\cal C}$.
\end{prop}

The notations $M_a$ and $\ell_{\bf s}$ used in Eq.\eqref{eqn:pC_from_values}
were introduced in Eq.\eqref{eqn:projectors_sum}.
Recall that, by deterministic, we mean $p_C({\bf s} | \nu) \in \{0, 1\}$.

\begin{proof}
First, let us prove that the HVM $(S, q_\rho, p)$ defined in this proposition
produces the same prediction as quantum mechanics.
By Eq.\eqref{eqn:projectors_sum}, the probability of an outcome ${\bf s} = (s_1, \dots, s_m)$
when measuring $\{T_{\bf a_1}, \dots, T_{\bf a_m}\}$ is given by 
$$
\Tr(\Pi_{\bf a}^{\bf s} \rho) = \frac{1}{|M_{\bf a}|} \sum_{{\bf u} \in M_{\bf a}}  \omega^{\ell_s({\bf u})} \Tr(T_{\bf u} \rho).
$$
Replacing $\Tr(T_{\bf u} \rho)$ by its value in terms of the value assignments, we obtain
$$
\Tr(\Pi_{\bf a}^{\bf s} \rho) = \sum_{\nu \in S} \left( \frac{1}{|M_{\bf a}|} \sum_{{\bf u} \in M_{\bf a}}  \omega^{\ell_{\bf s}({\bf u})} \cdot \lambda_\nu({\bf u}) \right) q_{\rho}(\nu).
$$
This proves that the HVM $(S, q_\rho, p)$ defined by Eq.\eqref{eqn:pC_from_values}
reproduces the quantum mechanical predictions.

It is deterministic since the sum 
$\sum_{u \in M_a}  \omega^{\ell_s({\bf u})} \cdot \lambda_\nu({\bf u})$,
which is the sum of the values of a character, is either 0 or $|M_{\bf a}|$,
implying that $p_C({\bf s} | \nu) \in \{0, 1\}$.
This HVM is also non-contextual. Indeed, conditional probabilities are defined in 
such a way that they do not depend on the particular choice of
generators ${\bf a_i}$ for the subspace $M_{\bf a}$, that is they do not depend
on the context.
\end{proof}

We now prove the converse statement.
Together with the non-contextuality assumption, determinism of the HVM yields extra 
compatibility contraints on the functions $\alpha_\nu$.
The following proposition proves that $\alpha_\nu$ is completely 
determined by its value $\alpha_\nu(\{T_{\bf u}\})$ over single operator contexts $\{T_{\bf u}\}$.
We shorten the notation $\alpha_\nu(\{T_{\bf u}\})$ by $\alpha_\nu({\bf u})$.
Moreover, we show that $\alpha_\nu$ is additive when $[{\bf u, v}] = 0$.
Then, we will prove that we can construct a set of NC value assignments from 
the maps $\alpha_\nu$.

\begin{prop} \label{prop:sharpness_implies_additivity}
If $(S, p, \alpha)$ is a deterministic NCHVM then, 
\begin{itemize}
\item for all context $C = \{T_{\bf a_1}, \dots, T_{\bf a_m} \} \in \cal C$, we have 
 $$
 \alpha_\nu(\{T_{\bf a_1}, \dots, T_{\bf a_m}\}) = (\alpha_\nu({\bf a_1}), \dots, \alpha_\nu({\bf a_m})),
 $$
\item if $T_{\bf u}$ and $T_{\bf v}$ commute, \emph{i.e.} if $[{\bf u, v}] = 0$, we have
$$
\alpha_\nu({\bf u+v}) = \alpha_\nu({\bf u}) + \alpha_\nu({\bf v}) \cdot
$$
\end{itemize}
\end{prop}

\begin{proof}
To prove the first item, we consider two implementations of the measurement 
$(\Pi_{\bf a_i}^s)_{s \in \Z_d}$ for some ${\bf a_i} \in V$.
First, one can simply measure $\{T_{\bf a_i}\}$ and reveal the outcome $s_i$.
A second implementation is obtained via the context 
$C = \{T_{\bf a_1}, \dots, T_{\bf a_m} \} \in \cal C$ and the map
$\sigma_C$ that sends a measurement outcome ${\bf t} = (t_1, \dots, t_m)$
onto its $i$-th component $t_i$. In other words, we measure these $m$
operators but we only keep the outcome of $T_{\bf a_i}$.
By non-contextuality, these two procedures yield the same conditional
probabilities at the level of the HVM, that is
$
p_{C}(\sigma_C^{-1}(s) | \nu) = p_{T_{\bf a_i}}(s | \nu)
$
for all $\nu \in S$, for all $s \in \Z_d$.
Replacing the first term by its definition, we obtain
\begin{align} \label{eqn:proof:multiplicativity}
\sum_{\substack{(t_1, \dots, t_m) \in \Z_d^m\\ t_i = s}} p_{C}({\bf t} | \nu) = p_{T_{\bf a_i}}(s | \nu).
\end{align}
Fix $\nu \in S$ and denote $\alpha_\nu(C) = (x_1, \dots, x_m)$.
Sharpness of the measurements implies that
$
p_{C}({\bf t} | \nu) = \prod_{j=1}^m \delta_{x_j, t_j}
$
and 
$
p_{T_{\bf a_i}}(s | \nu) = \delta_{\alpha_\nu(\bf a_i), s}.
$
Injecting these expressions in Eq.\eqref{eqn:proof:multiplicativity} produces the
equality
\begin{align} \label{eqn:proof:multiplicativity2}
\sum_{\substack{(t_1, \dots, t_m) \in \Z_d^m\\ t_i = s}} \prod_{j=1}^m \delta_{x_j, t_j} = \delta_{\alpha_\nu({\bf a_i}), s}.
\end{align}
that is satisfied for all $s \in \Z_d$.
The only possibility to have a non-trivial product at the left-hand side is to pick
$t_j = x_j$ for all $j \neq i$, leading to
$$
\delta_{x_i, s} = \delta_{\alpha_\nu({\bf a_i}), s}.
$$
This equality is satisfied for all $s \in \Z_d$, proving that
$\alpha_\nu({\bf a_i}) = x_i$.
This concludes the proof of the first property.

\medskip
The second item is proven using two implementations of the measurement
of $(\Pi_{\bf u+v}^s)_{s \in \Z_d}$.
First, we consider the direct implementation by measuring $T_{\bf u+v}$.
Then, we use the context $C = \{T_{\bf u}, T_{\bf v}\}$ with the postprocessing map
$\sigma_C(s_{\bf u}, s_{\bf v}) = s_{\bf u} + s_{\bf v}$.
Non-contextuality leads to 
$$
\sum_{k \in \Z_d} p_{C}((k, s-k) | \nu) = p_{T_{\bf u+v}}(s | \nu).
$$
Using the first result, the delta function describing the conditional
distribution for $C$ is associated with
$
\alpha_\nu(\{T_{\bf u}, T_{\bf v}\}) = (\alpha_\nu({\bf u}), \alpha_\nu({\bf v})).
$
This implies
$$
\sum_{k \in \Z_d} \delta_{\alpha_\nu({\bf u}), k} \cdot \delta_{\alpha_\nu({\bf v}), s-k} 
= \delta_{\alpha_\nu({\bf u+v}), s}.
$$
The left-hand side is equal to $\delta_{\alpha_\nu({\bf v}), s-\alpha_\nu({\bf u})}$
proving that $\alpha_\nu({\bf u}) + \alpha_\nu({\bf v}) = \alpha_\nu({\bf u+v})$.
\end{proof}

As a corollary, we prove that the maps ${\bf u} \mapsto \alpha_\nu({\bf u})$
define a family of NC value assignments. This complete the proof of the equivalence
between (ii) and (iii).

\begin{coro}
If $(S, q_\rho, \alpha)$ is a deterministic NCHVM then
the triple $(S, q_\rho, \lambda)$ where the map $\lambda_\nu: V \rightarrow U_d$ 
is defined by
$
\lambda_\nu({\bf u}) = \omega^{\alpha_\nu({\bf u})}
$
for all $\nu$, 
is a set of NC value assignments.
\end{coro}

\begin{proof}
Additivity of the maps $\lambda_\nu$ was proven in Proposition~\ref{prop:sharpness_implies_additivity}.
It remains to prove that this value assignment provides a good prediction for the expectation of 
operators $T_{\bf u}$.
Writing this operator as a linear combination of projectors $T_{\bf u} = \sum_{s \in \Z_d} \omega^s \Pi_{\bf u}^s$,
we find
$$
\Tr(T_{\bf u} \rho) = \sum_{s \in \Z_d} \omega^s \Tr(\Pi_{\bf u}^s \rho).
$$
Using the prediction of the HVM
and sharpness of the measurements, we 
obtain
\begin{align*}
\Tr(T_{\bf u} \rho) 
= \sum_{\nu \in S} \left( \sum_{s \in \Z_d} \omega^s \delta_{\alpha_\nu({\bf u}), s} \right) q_\rho(\nu)
= \sum_{\nu \in S} \omega^{\alpha_\nu({\bf u})} q_\rho(\nu)
\end{align*}
proving the Corollary.
\end{proof}

\section{Acknowlegdments}
The authors thank Ana Bel\'en Sainz,  Mark Howard and Joel Wallman for discussions.
ND acknowledges funding provided by the Institute for Quantum Information and Matter, 
an NSF Physics Frontiers Center (NSF Grant PHY-1125565) with support of 
the Gordon and Betty Moore Foundation (GBMF-2644).
CO  acknowledges funding from NSERC.
JBV acknowledges funding from SIQS, AQuS.
RR is supported by NSERC and Cifar.

\bibliographystyle{unsrt}
\bibliography{contextuality.bib}

\end{document}